\relax
\RequirePackage{etex}
\documentclass[a4paper, 11pt]{article} 
\usepackage[margin=1in]{geometry}
\usepackage[hyphens]{url}
\usepackage{graphicx}
\usepackage{caption}
\usepackage{titlesec}
\titleformat{\section}{\normalfont\bfseries\large\raggedright}{\thesection.\ }{0.10cm}{}
\titlespacing*{\section}{0pt}{0.24in plus .8ex}{0.10in plus .2ex}
\titleformat{\subsection}{\normalfont\bfseries\normalsize\raggedright}{\thesubsection.\ }{0.08cm}{}
\titlespacing*{\subsection}{0pt}{0.2in plus .8ex}{0.08in plus .2ex}

\usepackage{natbib}
\bibpunct{(}{)}{;}{a}{,}{,}

\usepackage{hyperref}
\hypersetup{colorlinks,linkcolor=blue,citecolor=blue,
    urlcolor=magenta,linktocpage,plainpages=false}

\usepackage{amssymb,amsmath,amsthm,microtype,booktabs,enumitem,tikz,subcaption,mathtools,array,ifthen,pgf,bbm,comment}
\usepackage[export]{adjustbox}
\usepackage[utf8]{inputenc}
\usetikzlibrary{arrows,positioning,backgrounds,automata}

\newcommand{\reals}{\mathbb{R}}
\newcommand{\X}{\mathbb{X}}
\newcommand{\Y}{\mathbb{Y}}
\newcommand{\Z}{\mathbb{Z}}
\newcommand{\US}{\mathbb{S}}
\newcommand{\N}{\mathbb{N}}
\renewcommand{\P}{\mathbb{P}}

\newcommand{\One}{\mathbbm{1}}
\newcommand{\x}{\textbf{x}}
\newcommand{\y}{\textbf{y}}

\renewcommand{\bar}{\overline}
\newcommand{\Trans}{\intercal}
\DeclareMathOperator*{\relintnp}{relint}
\newcommand{\relint}[1]{\operatorname{\relintnp}(#1)}
\DeclareMathOperator*{\argmax}{arg\,max}
\DeclareMathOperator*{\expnp}{exp}
\newcommand{\ex}[1]{\operatorname{\expnp}(#1)}
\newcommand\restr[2]{{
  \left.\kern-\nulldelimiterspace 
  #1 
  \vphantom{\big|} 
  \right|_{#2} 
  }}

\newtheorem{definition}{Definition}
\newtheorem{theorem}{Theorem}
\newtheorem{corollary}{Corollary}
\newtheorem{lemma}{Lemma}
\newtheorem{proposition}{Proposition}

\title{No-Regret Learning in Games is Turing Complete}

\author{Gabriel P. Andrade\textsuperscript{\rm 1} \and Rafael Frongillo\textsuperscript{\rm 1} \and Georgios Piliouras\textsuperscript{\rm 2}\\}
\date{\textsuperscript{\rm 1} Department of Computer Science\\
University of Colorado Boulder\\
\{gabriel.andrade~;~raf\}@colorado.edu\\
\vskip8pt
\textsuperscript{\rm 2} Engineering Systems and Design\\
Singapore University of Technology and Design\\
georgios@sutd.edu.sg}

\begin{document}

\maketitle

\begin{abstract}
Games are natural models for  multi-agent machine learning settings, such as generative adversarial networks~(GANs).
The desirable outcomes from algorithmic interactions in these games are encoded as game theoretic equilibrium concepts, e.g.~Nash and coarse correlated equilibria.
As directly computing an equilibrium is typically impractical, one often aims to design learning algorithms that iteratively converge to equilibria.
A growing body of negative results casts doubt on this goal, from non-convergence to chaotic and even arbitrary behaviour.
In this paper we add a strong negative result to this list: learning in games is Turing complete.
Specifically, we prove Turing completeness of the replicator dynamic on matrix games, one of the simplest possible settings.
Our results imply the undecicability of reachability problems for learning algorithms in games, a special case of which is determining equilibrium convergence.
\end{abstract}

\section{Introduction}\label{sec:intro}
Many multi-agent machine learning settings can be modeled as games, from social or economic systems with algorithmic decision-makers to popular learning architectures such as generative adversarial networks~(GANs).
Desired outcomes in these settings are often encoded as equilibrium concepts, and therefore a primary goal is identifying machine learning algorithms with provable convergence to these equilibria.

While there has been progress in deriving strong time-average convergence guarantees for popular online learning algorithms, the per-iteration behaviour of learning in games remains elusive.
Recent results attempt to formalize how elusive these dynamics can be, from non-convergence results to establishing chaotic, or even essentially arbitrary, behaviour~\citep{andrade2021learning,Benam2012PerturbationsOS,flokas2020no,giannou21survival,letcher2021impossibility}.
Experiments confirm that chaos can actually be typical behaviour~\citep{sanders2018prevalence}.

In this work, we add an even more sobering negative result to this list: learning in games is Turing complete.
Specifically, we show that replicator dynamics in matrix games, one of the simplest possible settings, can simulate an arbitrary Turing machine (Theorem~\ref{thm:main}).
Here simulation is defined in terms of reachability, a natural decision problem for dynamical systems that asks whether a given system and initial condition eventually intersects (reaches) a certain set; a dynamical system simulates a Turing machine if the corresponding halting problem reduces to the reachability problem.
Our proof combines two recent results, on the Turing completeness of fluid dynamics~\citep{cardona2021turing}, and on the approximate universality of learning in games~\citep{andrade2021learning}.

We believe our results have far-reaching implications for the literature on learning in games.
Most immediate is the fact that the reachability problem is undecidable for no-regret learning in general~(Corollary~\ref{cor:genral_undecidable}).
This result calls into question the feasibility of equilibration as a goal, since even deciding whether a learning algorithm gets close to an equilibrium is undecidable.
More broadly, these results establish the computational power of learning dynamics in games---and accordingly, their inherent complexity as formalized by computabiity theory.

Beyond the continuous-time setting,
we borrow tools from numerical analysis to show that the multiplicative weights algorithm can simulate any bounded Turing machine~(Theorem~\ref{thm:bounded_time}).
Extending this analysis to arbitrary Turing machines, and thus establishing Turing completeness for the discrete-time setting, may not be possible with the techniques we consider.
Establishing (or refuting) the Turing completeness of multiplicative weights is therefore left as an important open question, and one that will likely require entirely new techniques.

\section{Preliminaries}\label{sec:prelim}

\subsection{Matrix Games}
A \emph{finite $n$-player normal form game} consists of $n$ agents $[n] = \{1, \dots, n\}$, where each agent $i \in [n]$ can choose actions from a finite action set $S_i$.
Actions are chosen by agent $i$ according to a \emph{mixed strategy}, a distribution $\x_i$ in the probability $|S_i|$-simplex $\Delta^{|S_i|} = \{\x_i \in \reals^{|S_i|}_+ : \sum_{s \in S_i} x_{i s} = 1\}$.
In normal form games, agents receive payoffs from pairwise interactions according to payoff matrices $A_{i,j} \in \reals^{|S_i| \times |S_j|}$ where $i,j \in [n]$ and $i \neq j$.
Given that mixed strategies $\x_i \in \Delta^{|S_i|}$ and $\x_j \in \Delta^{|S_j|}$ are chosen, agent $i$ receives payoff $\x^\Trans_i A_{i,j} \x_j$.
These payoffs yield a natural optimization problem for each agent, where agents act strategically and independently to maximize their expected payoff over the other agents' mixed strategies, i.e.
\begin{equation}\label{eq:game_opt}
    \max\limits_{\x_i \in \Delta^{|S_i|}} \sum_{j \in [n];~j \neq i} \x_i^\Trans A_{i,j} \x_j, \qquad i \in [n]~.
\end{equation}
Throughout the paper we'll restrict our attention to the case known as \emph{matrix games}, when $n=2$.

\subsection{Follow-the-Regularized-Leader~(FTRL) Learning and Replicator Dynamics}\label{subsec:learning_prelim}
In many game settings, the optimization in eq.~\eqref{eq:game_opt} is a moving target since the opponent adaptively updates their strategy and the payoff matrix may be unknown.
In such settings, arguably the most well known class of algorithms is Follow-the-Regularized-Leader (FTRL). 
The continuous-time version of an FTRL algorithm is as follows.
Given initial payoff vector $\y_i(0)$, an agent $i$ that plays against agent $j$ in a matrix game $A_{i,j}$ updates their strategy at time $t$ according to
\begin{equation}
\label{eqn:FTRL}
\begin{aligned}
\y_i(t)&= \y_{i}(0)+\int_0^t A_{i,j}\x_{j}(s) ds \\
\x_i(t)&= \argmax_{\x_i\in \Delta^{|S_i|}} \{\langle \x_i, \y_i(t)\rangle-h_i(\x_i)\} 
\end{aligned}
\end{equation}
where $h_i$ is strongly convex and continuously differentiable. FTRL effectively performs a balancing act between exploration and exploitation.
The cumulative payoff vector $\y_i(t)$ indicates the total payouts until time $t$, i.e.~if agent $i$ had played strategy $s_i \in S_i$ continuously from $t=0$ until time $t$, agent $i$ would receive a total reward of $\y_{is_i}(t)$. 
The two most well-known instantiations of FTRL dynamics are the online gradient descent algorithm when $h_i(\x_i)=||\x_i||_2^2$,
and the replicator dynamics (the continuous-time analogue of Multiplicative Weights Update~\citep{Arora05themultiplicative}) when $h_i(\x_i)=\sum_{s_i\in {\cal S}_i} \x_{is_i}\ln \x_{is_i}$. 
FTRL dynamics in continuous time has bounded regret in arbitrary games~\citep{mertikopoulos2018cycles}. 
For more information on FTRL dynamics and online optimization, see \cite{Shalev2012}.

In this paper, we will focus on \emph{replicator dynamics} (RD) as the learning process generating game dynamics. 
In addition to its role in optimization, replicator dynamics is the prototypical dynamic studied in evolutionary game theory~\citep{Weibull,Sandholm10} and is one of the key mathematical models of evolution and biological competition~\citep{Schuster1983533,taylor1978evolution}. 
In this context, replicator dynamics can be thought of as a normalized form of ecological population models, and is studied given a single payoff matrix $A$ and a single probability distribution $\x$ that can be thought abstractly as capturing the proportions of different species/strategies in the current population. 
Species/strategies get randomly paired up and the resulting payoff determines which strategies will increase/decrease over time.

Formally, the dynamics are as follows.
Let $A \in \reals^{m \times m}$ be a matrix game and $\x \in \Delta^m$ be the mixed strategy played.
RD on $A$ are given by:
\begin{equation} \label{eq:replicator}
    \dot{x}_i = \frac{d x_i}{d t} = x_i \left((A\x)_i - \x^\Trans A \x \right), \qquad i \in [n]
\end{equation}
Under the symmetry of $A_{i,j}=A_{j,i}$, and of initial conditions (i.e. $\x_i=\x_j$), it is immediate to see that under the $\x_i, \x_j$ solutions of eq.~\eqref{eqn:FTRL} are identical to each  other and to the solution of eq.~\eqref{eq:replicator} with $A=A_{i,j}=A_{j,i}$. 
For our purposes, it will suffice to focus on exactly this setting of matrix games defined by a single payoff matrix $A$ and a single probability distribution $\x$, which is actually the standard setting within evolutionary game theory.

\subsection{Dynamical Systems Theory}\label{subsec:dynamics_prelim}
A dynamical system is a mathematical model of a time-evolving process.
The objects undergoing change in a dynamical system is called its \emph{state} and is often denoted by $\x \in \X$, where $\X$ is a topological space called a \emph{state space}. 
For most of this paper we will be focusing on \emph{continuous time} systems, but in~\S\ref{sec:discretization} we will consider \emph{discrete time} systems derived from numerical approximations of their continuous counterpart.
To distinguish between continuous and discrete time, we will use $\x(t)$ to describe the state as a function of continuous time $t \in \reals$ and $\x^t$ to describe the state as a function of discrete time $t \in \Z$.

Change between states in a continuous time dynamical system is described by a \emph{flow} $\Phi: \X \times \reals \to \X$ satisfying two properties:
\begin{enumerate}[label=(\roman*)]
    \item For each $t \in \reals$, $\Phi(\cdot,t): \X \to \X$ is bijective, continuous, and has a continuous inverse.
    \vspace*{-2pt}
    \item For every $s,t \in \reals$ and $\x \in \X$, $\Phi(\x,s+t) = \Phi(\Phi(\x,t),s)$.
\end{enumerate}
Intuitively, flows describe the evolution of states in the dynamical system.
Given a time $t \in \reals$, the flow gives us the relative movement of every point $\x \in \X$; we will denote this by the map $\Phi^t : \X \to \X$.
Similarly, given a point $\x \in \X$, the flow captures the trajectory of $\x$ as a function of time; in an abuse of notation, we will denote this by $\x(t)$ where $t$ is changing.

Continuous time dynamical systems are often given as systems of \emph{ordinary differential equations} (ODEs).
Systems of ODEs describe a \emph{vector field} $V:\X \to T\X$ which assigns to each $\x \in \X$ a vector in the tangent space of $\X$ at $\x$.
The unit sphere $\US^{n} = \{\x \in \reals^{n+1} : \|x\|_2^2 = 1\}$ will play a special role in proving Theorem~\ref{thm:main}, in which case the tangent space $T\US^n$ at each $\x \in \US^n$ is $\{\y \in \reals^n : \x \cdot \y = 0\}$.
Intuitively, the tangent space defines bundles of vectors that ensure the system's states remain well defined on the state space as time progresses.
A system of ODEs is said to \emph{generate} (or~\emph{give}) a flow $\Phi$ if $\Phi$ describes a solution of the ODEs at each point $\x \in \X$. 
Throughout this paper we assume that all dynamical systems discussed can be given by a system of ODEs.
For this reason, we will use the term \emph{dynamical system} to refer to the system of ODEs, the associated vector field, and a generated flow interchangeably.
A well known result in dynamical systems theory states that, for Lipschitz-continuous systems of ODEs, the generated flow is unique (see~\cite{perko1991differentialeqsbook,meiss2007differentialbook}) and using these terms interchangeably is well defined.

An important notion for proving Theorem~\ref{thm:main}, and dynamical systems in general, is that of a \emph{global attracting set} of the dynamical system.
Let $\Phi$ be a flow generated by some dynamical system on $\X$.
We say $\Y \subset \X$ is \emph{forward invariant} for the flow $\Phi$ if $\Phi^t(\y) \in \Y$ for every $t \geq 0$, $\y \in \Y$.
We say $\Y \subset \X$ is \emph{globally attracting} for the flow $\Phi$ if $\Y$ is nonempty, forward invariant, and
\begin{equation}
  \Y \supseteq \bigcap\limits_{t > 0} \{\Phi^t(\x) : \x \in \X\}~.
\end{equation}
Stated informally, if $\Y$ is globally attracting it will eventually capture the dynamics of $\Phi$ starting from any point in $\X$ after some transitionary period of time.

Now let $\X$ and $\Y$ be two topological spaces.
We say that a function $f: \X \to \Y$ is a \emph{homeomorphism} if (i) $f$ is bijective, (ii) $f$ is continuous, and (iii) $f$ has a continuous inverse.
Furthermore, two flows $\Phi: \X \times \reals \to \X$ and $\Psi: \Y \times \reals \to \Y$ are \emph{homeomorphic} if there exists a homeomorphism $g: \X \to \Y$ such that for each $\x \in \X$ and $t \in \reals$ we have $g(\Phi(\x,t)) = \Psi(g(\x),t)$.
If $g$ is also $C^1$ and has a $C^1$ inverse, then we say $g$ is a \emph{diffeomorphism} and that the flows $\Phi$ and $\Psi$ are \emph{diffeomorphic}.
Observe that every diffeomorphism is also a homeomorphism, and thus every pair of diffeomorphic flows are also homeomorphic.
Homeomorphic (resp.~diffeomorphic) flows satisfy a strong, and typical, notion of equivalence between dynamical systems.
Intuitively, two dynamical systems are homeomorphic if their trajectories can be mapped to one another by stretching and bending space.

\subsection{Turing Machines}\label{subsec:TM_prelim}
Throughout this paper we rely crucially on the notion of a Turing complete dynamical systems, i.e.~a dynamical system able to simulate any Turing machine.
We will briefly recall the Turing machine model and formalize its relationship with dynamical systems.

A \emph{Turing machine} is given by a tuple $T = \left(Q,\Sigma,\delta,q_0,q_{\text{halt}} \right)$ where
\begin{itemize}
    \item $Q$ is a finite \emph{set of states}, including an \emph{initial state} $q_0$ and a \emph{halting state} $q_{\text{halt}}$;
    \item $\Sigma$ is an \emph{alphabet} with cardinality at least two;
    \item $\delta: Q \times \Sigma \to Q \times \Sigma \times \{-1,0,1\}$ is a \emph{transition function}.
\end{itemize}
For a given Turing machine $T$ and an \emph{input tape} $s = (s_i)_{i \in \Z} \in \Sigma^{\Z}$, the Turing machine's computation is carried out according to the following process:
\begin{itemize}[leftmargin=4em,itemsep=0pt]
    \item[{[0]}] Initialize the \emph{current state} $q$ to $q_0$, and the current tape $w = (w_i)_{i \in \Z}$ to be $s$.
    \item[{[1]}] If $q = q_{\text{halt}}$ then halt the algorithm and return $w$ as output.
    Otherwise compute $\delta(q,w_0) = (q',w'_0,\sigma)$, where $\sigma \in \{-1,0,1\}$.
    \item[{[2]}] Update the current state and tape by setting $q = q'$ and the $0^{\text{th}}$ position of $w$ to $w_0 = w'_0$.
    \item[{[3]}] Update $w$ with the $\sigma$ shifted tape $(w_{i + \sigma})$, then return to [1].
\end{itemize}

Without loss of generality, we will assume that Turing machines adhere to standard simplifying conventions (cf.~\cite{sipser1996intro}).
Specifically, we assume that the alphabet $\Sigma = \{0,\dots,9\}$ and any given tape of the Turing machine only has a finite number of symbols different from $0$, where $0$ represents the special ``blank symbol''.
Under these assumptions it follows that there exists a finite (possibly large) integer $k_0 > 0$ such that any tape $w$ satisfies
\begin{equation}\label{eq:finite_tape}
   w = (w_i)_{i \in \Z} = \dots 0 0 0 w_{-k_0} \dots w_{k_0} 0 0 0 \dots
\end{equation}
with each $w_i \in \Sigma$.
Equivalently, at any given step in the Turing machine's evolution, these assumptions ensure there can be at most $2 k_0 + 1$ non-blank symbols on the tape.
In particular, we get that the space of configurations of a Turing machine $T$ is $Q \times A \subset Q \times \Sigma^{\Z}$, where $A$ is the subset of strings taking the form~\eqref{eq:finite_tape}.

The construction of dynamical systems that simulate Turing machines is at the heart of our results, and has been studied for various problems in physics~\citep{Reif1994ComputabilityAC,Freedman98pnp,cardona2021constructing}.
Although equivalent definitions exist, our analyses will adopt the formalisms used by recent work on fluid dynamics~\citep{cardona2021turing,tao2017universality}.
An analogous definition can be given for flows on a manifold.

\begin{definition}\label{def:turing_sim}
A vector field $X$ on a manifold $M$ \emph{simulates} a Turing machine $T$ if there exists an explicitly constructible open set $U_{w_{-k}, \dots, w_k} \subset M$ corresponding to each finite string $w_{-k}, \dots, w_k \in \Sigma$, and an explicitly constructible point $p_s \in M$ corresponding to each $s \in \Sigma^{\Z}$, such that: 
$T$ with input tape $s$ halts with an output tape having values $w_{-k}, \dots, w_k$ in positions $-k, \dots, k$ respectively if and only if the trajectory of $X$ through $p_s$ intersects $U_{w_{-k}, \dots, w_k}$.
\end{definition}

Intuitively, a dynamical system simulates a Turing machine if there is a correspondence between trajectories reaching certain sets and computations halting with certain configurations.
In particular, constructing the point $p_s$ depends only on the Turing machine $T$ and input tape $s$, while constructing the set $U_{w_{-k}, \dots, w_k}$ depends only on the specified halting configuration of $T$.
Both here and throughout the paper, we say a mathematical object (e.g.~points, sets, or matrices) is \emph{constructible} if it can be computed in finite time; constructability is not explicitly used in our arguments, but is important for nuanced technical reasons since it disallows pathological scenarios such as having all information about a machine's computations being encoded in an initial condition.    

Definition~\ref{def:turing_sim} leads to a natural notion of Turing completeness for dynamical systems. 
\begin{definition}\label{def:turing_comp}
A dynamical system is \emph{Turing complete} if it can simulate a universal Turing machine $T$. 
\end{definition}

\section{Turing Complete Dynamics on Matrix Games}\label{sec:RD_TC}
Our goal in this section is to establish the Turing completeness of replicator dynamics; in~\S\ref{subsec:poly_embed} we provide all precursory results required to prove the main result in~\S\ref{subsec:main_RD}.

\subsection{Turing Complete Vector Fields and Approximation-Free Game-Embeddings}\label{subsec:poly_embed}
Our construction of Turing complete game dynamics relies crucially on the notion of \emph{generalized Lotka-Volterra} (GLV) vector fields.
In particular, two properties of GLV vector fields will play a key role in the proof: (i)~\emph{polynomial} vector fields on $\reals^n_{++}$ are a special case of GLV vector fields, and (ii)~GLV vector fields can be embedded into RD on a matrix games \emph{without} approximation.

Formally, a GLV vector field is a vector field on $\reals^n_{++}$ given by the system of ODEs
\begin{equation} \label{eq:general_LV}
    \dot{x}_i = \frac{d x_i}{d t} = x_i \left(\lambda_i +  \sum_{j \in [m]} A_{ij} \prod_{k \in [n]} x^{B_{jk}}_k \right), \qquad i \in [n]
\end{equation}
where $m$ is some positive integer, $\lambda \in \reals^n$, $A \in \reals^{n \times m}$, and $B \in \reals^{m \times n}$.
Since exponents given by $B$ can be any real number, the terms in the parentheses are \emph{multivariate generalized polynomials}.
In special cases where the ODEs are standard multivariate polynomials, GLV vector fields equate to polynomial vector fields---a fact straightforwardly shown by noting that any polynomial vector field $P = \{P_i\}_{i \in [n]}$ on $\reals^n_{++}$ is equivalent to the GLV vector field $\tilde{P} = \{x_i(\tfrac{1}{x_i} P_i)\}_{i \in [n]}$.

Polynomial and GLV vector fields play an integral role by allowing us to invoke recent results by~\cite{cardona2021turing} and~\cite{andrade2021learning}.
The starting point of our construction can stated as follows:
\begin{proposition}[Theorem 4.1 of~\cite{cardona2021turing}]\label{prop:cardona}
There exists a constructible polynomial vector field $X$ of degree $58$ on $\US^{17}$ which is Turing complete and bounded.
\end{proposition}

In Appendix~\ref{a:cardona} we provide a proof sketch of this result;  we refer the reader to~\cite{cardona2021turing} for the full proof.
In~\S\ref{subsec:main_RD} we will extend the Turing completeness from Proposition~\ref{prop:cardona} to replicator dynamics in matrix games by leveraging recent work by~\cite{andrade2021learning}.
In essence, \cite{andrade2021learning} showed that GLV vector fields can approximate essentially any dynamical system, and that any GLV vector field can be embedded into the dynamics of RD on some matrix game.
In this paper we only rely on the latter result, since polynomial vector fields are already a special case of GLV vector fields and thus do not need to be approximated.

\begin{proposition}[Theorem 3 of \cite{andrade2021learning}]\label{prop:glv}
Let $\tilde{P}$ be a GLV vector field on $\reals^n_{++}$ and $\Phi$ be the flow generated by $\tilde{P}$.
For $m \geq n$, there exists a flow $\Theta$ on $\relint{\Delta^m}$ and a constructible diffeomorphism $f:\reals^n_{++} \to \P \subseteq \relint{\Delta^m}$ such that:
\begin{enumerate}[label=(\roman*)]
    \item The flow $\Theta$ on $\relint{\Delta^m}$ is given by RD on a matrix game with payoff matrix $A \in \reals^{m \times m}$.
    \item The flow $\restr{\Theta}{\P} = f(\Gamma)$ and $\Phi = f^{-1}(\restr{\Theta}{\P})$, where $\restr{\Theta}{\P}$ is the flow given by $\Theta$ restricted to $\P$.
    \item The integer $m-1$ is at least the number of unique monomials in $\tilde{P}$.
\end{enumerate}
\end{proposition}

At a high level, proving Proposition~\ref{prop:glv} boils down to composing an embedding trick introduced by~\cite{brenig1989universal} with Theorem~$7.5.1$ by~\cite{hofbauer1998book}.
The relationship highlighted here between $m-1$ and the number of monomials was not included in the original statement by~\cite{andrade2021learning}, however it is shown as part of an important step in their proof and is required for Corollary~\ref{cor:bounding_degree}.

\subsection{Replicator Dynamics on Matrix Games is Turing Complete}\label{subsec:main_RD}
To prove the main result of this section, Theorem~\ref{thm:main}, we will apply Proposition~\ref{prop:glv} on a diffeomorphism of the Turing complete vector field constructed in Proposition~\ref{prop:cardona}.

\begin{theorem}\label{thm:main}
There exists $m \geq 0$ and a constructible matrix game $A \in \reals^{m \times m}$ such that replicator dynamics on $A$ is Turing complete. 
\end{theorem}
\begin{proof}
Let $X$ be the Turing complete polynomial vector field on $\US^{17}$ given by Proposition~\ref{prop:cardona}.
We begin by embedding $X$ into a polynomial vector field $\bar{X}$ on $\reals^{18}$ where $\US^{17}$ is globally attracting.
Since trajectories of $\bar{X}$ are globally attracted to $\US^{17}$, a standard change of coordinates via translation yields a polynomial vector field that is well-defined on $\reals^{18}_{++}$.
Therefore, as polynomial vector fields on $\reals^{18}_{++}$ are a special case of GLV vector fields, we will conclude the proof by applying Proposition~\ref{prop:glv} from~\S\ref{subsec:poly_embed}.

Let $\{\phi_i\}_{i \in [18]}$ be the set of polynomials given by $X$.
Define $\pi(\x) = (1 - \|\x\|_2^2)$ for $\x \in \reals^{18}$.
Now define $\bar{X}$ as the vector field on $\reals^{18}$ given by the system 
\begin{align*}
    \dot{x}_i &= x_i \left(\pi(\x) + \frac{1}{x_i} \phi_i(\x)\right)\\
    &= x_i\pi(\x) + \phi_i(\x)~,
\end{align*}
for each $i \in [18]$.
By construction $\US^{17}$ is forward invariant under $\bar{X}$, as $\pi(\x) = 0$ on $\US^{17}$ and $X$ is forward invariant on $\US^{17}$.
Furthermore, observe that for $\x = \x(t) \in \reals^{18}$ the solutions of $\bar{X}$ satisfy
\begin{align*}
    \frac{d}{dt}\|\x\|_2^2 &= 2 \sum_{i \in [18]} x_i \dot{x}_i\\
    &= 2 \left( \sum_{i \in [18]} x_i^2 \pi(\x) + \sum_{i \in [18]} x_i \phi_i(\x) \right)\\
    &= 2 \pi(\x) \left(\sum_{i \in [18]} x_i^2 \right) + 2 \left(\sum_{i \in [18]} x_i \phi_i(\x) \right)\\
    &= 2 \pi(\x) \|\x\|_2^2\\
    &= 2 \|\x\|_2^2 \left(1 - \|\x\|_2^2\right)~,
\end{align*}
since, by definition of $T\US^{17}$, the constraint $\|\x\|_2^2 = 1$ ensures $X$ satisfies
\begin{equation*}
    2 \sum_{i \in [18]} x_i \phi_i(\x) = 0~.
\end{equation*}
The term $2 \|\x\|_2^2 \left(1 - \|\x\|_2^2\right)$ is a logistic equation in $\|\x\|_2^2$.
Thus, for every $\x \in \reals^{18}$, we know $\|\x\|_2^2 \to 1$ as $t \to \infty$.
It follows that $\US^{17}$ is globally attracting for the trajectories generated by $\bar{X}$.

Denote a standard translation of axes by $\sigma \in \reals$ as $F_\sigma: \reals^{18} \to \reals^{18}$, $\x \mapsto \x + \sigma\One$, where $\One$ is the all-ones vector.
Since solutions of $\bar{X}$ are attracted to $\US^{17}$ and Proposition~\ref{prop:cardona} ensures $\{\phi_i\}_{i \in [18]}$ is bounded due to the reparametrization done in~eq.~$(4.2)$ in~\cite{cardona2021turing}, there exists suitable values of $\sigma$ such that composing $F_\sigma$ with $\bar{X}$ yields a polynomial vector field that is forward invariant on $\reals^{18}_{++}$.
Formally, let $B > 0$ be the bound on $\{\phi_i\}_{i \in [18]}$ given in Proposition~\ref{prop:cardona}, i.e.\ for all $i \in [18]$ and $\x \in \reals^{18}$ the vector field $X$ satisfies $|\phi_i(\x)| \leq B$.
To ensure the translated vector field is forward invariant on $\reals^{18}_{++}$, it suffices to find $\sigma$ such that $Y = F_\sigma \circ \bar{X}$ is strictly positive on the boundary when $\y \in \reals^{18}_{++}$ has $y_i = 0$ for some $i \in [18]$.
By definition we know that $Y$ at any $\y \in \reals^{18}_{++}$ is identical to $\bar{X}$ at $\x = \y - \sigma\One$.
The system of equations $\{\dot{y}_i\}_{i \in [18]}$ is given by the system of equations $\{\dot{x}_i\}_{i \in [18]}$ under the substitution $\x = \y - \sigma\One$.
Therefore we find that, for $\y \in \reals^{18}_{++}$ with $y_i = 0$ for some $i \in [18]$, 
\begin{align*}
    \dot{y}_i 
    &= (y_i - \sigma) \pi(\y - \sigma\One) + \phi_i(\y - \sigma\One)\\
    &\geq (- \sigma) (1- \|\y - \sigma\One\|_2^2) - B\\
    &= \sigma \|\y - \sigma\One\|_2^2 - \sigma - B\\
    &\geq \sigma^3 - \sigma - B~,
\end{align*}
which implies $\dot{y}_i > 0$ whenever $B < -\sigma + \sigma^3$.
Thus, for values of $\sigma$ satisfying $B < -\sigma + \sigma^3$, we have $Y = F_\sigma \circ \bar{X}$ which is well defined on $\reals^{18}_{++}$ for all initial conditions in $\reals^{18}_{++}$.

By definition of $Y$, as a translated copy of $\bar{X}$, the set $F_\sigma(\US^{17})$ is globally attracting in $Y$, and $\restr{Y}{F_\sigma(\US^{17})}$ is a Turing complete polynomial vector field.
It follows we have constructed a polynomial vector on $\reals^{18}_{++}$ that inherits the Turing complete dynamics of $X$.
Since polynomial vector fields on $\reals^{18}_{++}$ are a special case of GLV vector fields on $\reals^{18}_{++}$, from Proposition~\ref{prop:glv} there exists a diffeomorphism $f: \reals^{18}_{++} \to \P \subseteq \relint{\Delta^m}$ from trajectories of $Y$ onto trajectories of an invariant submanifold of replicator dynamics on a matrix game $A \in \reals^{m \times m}$.

We conclude by showing how the Turing completeness of $X$ corresponds to Turing completeness for replicator dynamics on $A$.
Suppose we have a given Turing machine $T$, an input tape $s$, and some finite string $\omega$.
By Proposition~\ref{prop:cardona} there exists a point $p_s$ and open set $U_\omega$ such that trajectories of $X$ through $p_s$ intersect $U_\omega$ if and only if $T$ halts with input $s$ and output matching $\omega$ about the machine's head.
Our analysis above shows that $\restr{\bar{X}}{\US^{17}} = X$, so trajectories of $\bar{X}$ through $p_s$ intersect $U_\omega$ if and only if $T$ halts with input $s$ and output matching $\omega$.
Therefore, after translating $\bar{X}$, we know trajectories of $Y$ through $F_\sigma(p_s)$ intersect $F_\sigma(U_\omega)$ if and only if $T$ halts with input $s$ and output matching $\omega$.
Finally, since diffeomorphisms are closed under composition, we conclude that trajectories of replicator dynamics on $A$ through the point $f(F_\sigma(p_s))$ intersect the set $f(F_\sigma(U_\omega))$ if and only if $T$ halts with input $s$ and output matching $\omega$, where $f$ is the diffeomorphism above.
Thus, on an invariant submanifold of $\relint{\Delta^m}$, replicator dynamics on $A$ simulates $T$.
Taking $T$ to be a universal Turing machine completes the proof.
\end{proof}

An interesting corollary of Theorem~\ref{thm:main} is that we arrive at a bound on the number of actions needed for defining games where learning dynamics can be Turing complete.
The bound is likely loose for several reasons.
Firstly, the polynomial vector field from Proposition~\ref{prop:cardona} is not known to have minimal degree nor dimension.
Secondly, the combinatorial argument in~Appendix~\ref{a:cor1} makes no attempt at a nuanced count on the number of unique monomials in the polynomials given by these vector fields.
Deriving a tight bound is not only an interesting open question for game dynamics, but also for recent work in fluid dynamics~\citep{cardona2021turing,lizaur2021chaos} and analog computing~\citep{hainry2009decidability}.

\begin{corollary}\label{cor:bounding_degree}
For some $m \leq {76 \choose 18}+1$, there exists a matrix game $A \in \reals^{m \times m}$ such that replicator dynamics on $A$ is Turing complete.
\end{corollary}

\section{Undecidable Phenomena in No-Regret Learning Dynamics}\label{sec:undecidable}
The Turing completeness of replicator dynamics (i.e.~Theorem~\ref{thm:main}) has deep implications for machine learning and, more generally, learning in strategic environments.
Specifically, if a dynamical system simulates a Turing machine, Definition~\ref{def:turing_sim} gives a reduction from the halting problem for Turing machines to the reachability problem for dynamical systems, which we use alongside the Turing completeness established in Theorem~\ref{thm:main} to uncover the existence of undecidable reachability problems.
As will be discussed in~\S\ref{subsec:undecidable_imps}, the existence of undecidable problems makes it increasingly important that we understand computability in instances of reachability arising from fundamental solution concepts for game theory and machine learning.

\subsection{The Halting and Reachability Problems}\label{subsec:undecidable_probs}
The halting problem is a prototypical decision problem for Turing machines and is arguably the most famous undecidable problem in computer science.
Given a Turing machine $T$ and an input tape, the \emph{halting problem for $T$} asks whether or not $T$ will halt.
By contrast, the reachability problem is canonical for dynamical systems and has been studied in various control settings; given a dynamical system $X$ and a set of initial conditions, the \emph{reachability problem for $X$} asks whether or not $X$'s trajectory will intersect a predetermined set.
Although the computability of the halting problem is generally well understood in Turing machines, the computability of the reachability problem has not traditionally been studied in the context of game dynamics.
However, from the strong equivalence between halting and reachability required by Definition~\ref{def:turing_sim}, we immediately get a reduction between these classic decision problems.

\begin{proposition}\label{prop:reduction}
If a dynamical system $X$ simulates a Turing machine $T$, then the halting problem for $T$ reduces to the reachability problem for $X$.
\end{proposition}

The proof of this proposition follows directly from Definition~\ref{def:turing_sim}, since checking whether the dynamical system reaches a set becomes equivalent to checking whether the Turing machine halts by definition.
From Theorem~\ref{thm:main} we know that replicator dynamics on a matrix game can simulate a universal Turing machine.
Therefore, due to the undecidability of the halting problem in general, we deduce that the reachability problem can be undecidable for replicator dynamics on matrix games.

\begin{corollary}\label{cor:undecidable}
There exist matrix games where reachability is undecidable for replicator dynamics.
\end{corollary}

The corollary follows immediately from Proposition~\ref{prop:reduction} and Theorem~\ref{thm:main}, since the undecidability of the halting problem for universal Turing machines uncovers the undecidability of the reachability problem for replicator dynamics on matrix games.

\subsection{Implications for No-Regret Learning in Games}\label{subsec:undecidable_imps}
Games are primarily understood and studied via equilibrium concepts, e.g.~Nash equilibria, evolutionary stable strategies, and coarse correlated equilibria.
It is therefore unsurprising that the goal of learning in games is often to converge on some set of equilibria.
Yet, beyond certain special cases~(e.g.~potential games), learning behaviours remain largely enigmatic and there has been limited progress towards resolving non-convergence in general settings.
The results in this paper may explain why: determining convergence to a set of equilibria is a special case of reachability, and identifying learning algorithms that provably converge on such a set can be an undecidable problem even in very simple classes of games.
The goal of this section is to formalize this intuition.

In Corollary~\ref{cor:undecidable} we found that reachability can be undecidable for replicator dynamics on matrix games.
Therefore, taken as a negative result, Corollary~\ref{cor:undecidable} implies that undecidable trajectories can exist in larger classes of game dynamics where replicator dynamics on matrix games is a special case.
Unfortunately, replicator dynamics is special case of FTRL dynamics and no-regret learning dynamics more generally~\citep{mertikopoulos2018cycles}, which suggests these popular learning dynamics can inherit the negative result on any class of games containing matrix games.
Similarly, matrix games are very restricted and a special case of many popular classes of games, e.g.~normal form and smooth games.
As an example of how broadly these results generalize, matrix games in the FTRL framework describe quadratic objective functions and thus undecidable trajectories exist for optimization-driven learning over quadratic objectives.
Thus, as Corollary~\ref{cor:undecidable} holds for replicator dynamics on matrix games, we arrive at the reachability problem being generally undecidable for rich classes of game dynamics studied in the literature and used in practice.
\begin{corollary}\label{cor:genral_undecidable}
There exist games where reachability for no-regret learning dynamics is undecidable.
\end{corollary}

In light of Corollary~\ref{cor:undecidable}, the claim follows from our discussion above.
As determining convergence to sets of game theoretic solution concepts is a special case of the reachability problem, Corollary~\ref{cor:genral_undecidable} reveals that determining whether game dynamics converge to fundamental solution concepts is undecidable in general.
It is important to note that the undecidability may not hold for specific games or learning dynamics; the primary take-away is that undecidability is possible and has strong implications about how we should approach these important questions.

\section{Discrete Learning Dynamics and Turing Machine Simulations}\label{sec:discretization}
Thus far we have focused on the continuous-time replicator learning dynamics, but in practice discrete-time learning dynamics are typically used.
A folk result in the study of game dynamics states that the multiplicative weights update (MWU) algorithm is essentially an Euler discretization of replicator dynamics.
It is therefore natural to ask whether MWU, the discrete analogue of replicator dynamics, are also Turing complete.
Unfortunately, as will be shown in this section, standard numerical error analyses are likely insufficient for proving Turing completeness in discrete time; intuitively, the reason is because discretizations of a continuous time process will yield error bounds that grow as a function of time.
We will formalize these error bounds in~\S\ref{subsec:MWU} and use them in~\S\ref{subsec:discrete_sims} to begin untangling the computational power of MWU.
Discussions of related open questions are left for~\S\ref{sec:conclusion}.

\subsection{Discretization Error of Multiplicative Weights Updates}\label{subsec:MWU}
The fact that MWU is a discretization of replicator dynamics is well known in the field of game dynamics, but a precise derivation of this relationship is often omitted.
For clarity in our analysis of discretization errors, we will highlight one possible discretization that reveals MWU as a discrete-analogue of replicator dynamics in Appendix~\ref{a:MWU_to_rep}.
The discretization we arrive at is used to find a bound on the cumulative error of MWU relative to replicator dynamics, which is crucial for the analyses and discussion to follow.

\begin{lemma}\label{lem:LE}
Let $\Phi$ be the flow generated by replicator dynamics and $\x^{t}$ be the mixed strategy found on the $t^{\text{th}}$ iterate of MWU. 
The error accrued by a single iteration of MWU with step-size $\eta > 0$ is
\begin{equation*}
    \|\x^{t+1} - \Phi(\eta, \x^{t})\|_{\infty} \leq  1 - e^{-2\eta}~.
\end{equation*}
\end{lemma}
The proof of Lemma~\ref{lem:LE} consists of relatively straightforward calculations, but requires carefully handling nonlinearities introduced by MWU; a full proof is included in Appendix~\ref{a:LE}.

Using Lemma~\ref{lem:LE} as a basis, we can bound the error accrued over multiple iterations of MWU.\footnote{In the language of numerical analysis, Lemma~\ref{lem:LE} gives the local error used to find the global error in Lemma~\ref{lem:GE}.}

\begin{lemma}\label{lem:GE}
Let $\Phi$ be the flow generated by replicator dynamics and $\x^{t}$ be the mixed strategy found on the $t^{\text{th}}$ iterate of MWU.
The error accrued after $t+1$ iterations of MWU with step-size $\eta > 0$ is
\begin{equation*}
    \|\x^{t+1} - \Phi(t\eta, \x^{0})\| \leq \mathcal{O}\left(e^{t}\right)~.
\end{equation*}
\end{lemma}
A full derivation of Lemma~\ref{lem:GE} is found in Appendix~\ref{a:GE}, and follows from using Lemma~\ref{lem:LE} alongside standard techniques for bounding error in iterated numerical methods.

\subsection{Simulating Bounded Turing Machines with Multiplicative Weights Update}\label{subsec:discrete_sims}
The result in Lemma~\ref{lem:GE} shows that, relative to replicator dynamics, the error accrued by MWU will grow with the number of iterations.
Error growing as a function of time is problematic when simulating a Turing machine by using MWU as discretization of replicator dynamics.

Recall that in Theorem~\ref{thm:main} we showed that replicator dynamics can simulate a universal Turing machine because it can embed a dynamical system that simulates a universal Turing machine, which is done to ensure the Turing machine's halting remains equivalent to the dynamics' trajectories reaching a certain set.
However, in general, determining whether such a Turing machine will halt or how many steps are required to halt is undecidable.
Therefore, without an a priori bound on the maximum amount of time needed to determine whether the machine halts or not, we cannot choose step sizes for MWU that guarantee the discretization remains sufficiently close to replicator dynamics when intersecting the relevant set.

\begin{theorem}\label{thm:bounded_time}
Let $k > 0$ be a finite integer and $\mathbb{T}_{k}$ be the set of Turing machines that we can determine to halt or not after $k$ steps of computation.
For any $k$, there exists step sizes $\eta > 0$ such that MWU with step-size $\eta$ can simulate any Turing machine in $\mathbb{T}_k$. 
\end{theorem}

The result follows from the construction of the open sets used in Proposition~\ref{prop:cardona} and the fact that we can ensure MWU's discretization error stays sufficiently small over any finite window of time due to Lemma~\ref{lem:GE}.
Resolving the limitations of Theorem~\ref{thm:bounded_time}, and uncovering the true computational power of discrete algorithms such as MWU, will likely require new technical approaches for bounding errors or simulating Turing machines.

\section{Conclusion}\label{sec:conclusion}
We have shown that replicator dynamics in matrix games can simulate universal Turing machines.
In continuous time, this observation was extended to provide deeper insight into the complexities of game theoretic learning.
In fact, as highlighted in~\S\ref{sec:undecidable}, the plurality of negative results on game dynamics can be understood as a natural byproduct of Theorem~\ref{cor:genral_undecidable}.
Given that the present paper uses replicator dynamics specifically and matrix games broadly, complimenting the results given here with analyses based on other learning dynamics and classes of games could be instrumental in guiding future research by finding settings where designing well-behaved game dynamics is a tractable problem.
As was done for Turing machines in computational complexity theory and becomes more natural given the techniques used in our analyses, compartmentalizing the complexity of learning in games using traditional complexity classes suggests a promising line of investigation for finding tractable settings for learning in games.

\begin{figure}[t]
  \centering
  \includegraphics[width=\textwidth]{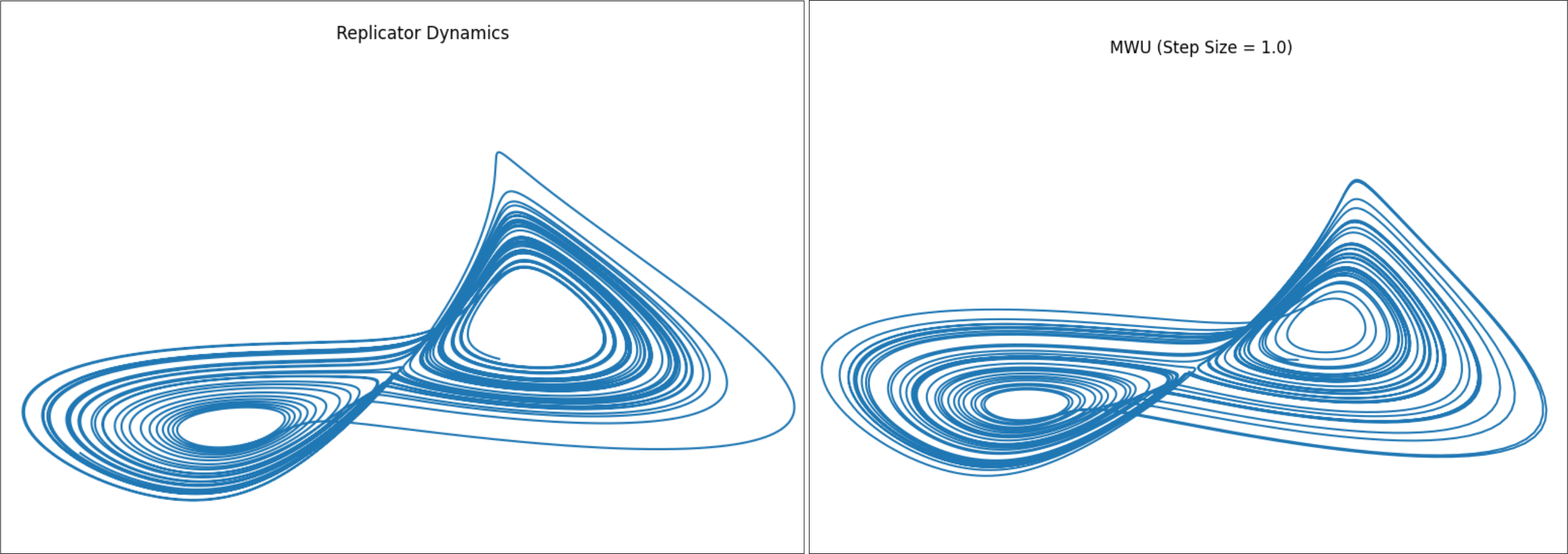}
  \caption{A comparison of replicator dynamics and MWU on a matrix game derived by~\cite{andrade2021learning} to simulate a chaotic dynamical system.
  On the left is replicator dynamics with the dynamics embedded into its behaviours, whereas on the right we have $10000$ iterations of MWU with a relatively large step size.
  Although not identical, it is clear that MWU retains the intricate and complex behaviours of replicator dynamics.}
  \label{fig:shadowing}
\end{figure}

In discrete time, the Turing completeness of replicator dynamics was used to show that MWU can simulate bounded Turing machines.
However, our approach does not disallow for the possibility of MWU being Turing complete as well; using MWU's relationship to replicator dynamics seems to have inherent numerical limitations arising from error growing with time. 
Since discrete-time learning is more applicable in practice, it remains an important open question to determine whether MWU and other discrete learning algorithms are Turing complete.
That being said, the smoothness constraints on continuous-time learning often leads to better behaved dynamics than discrete-analogues, and thus the study of continuous dynamics generally serves as restricted special case of what is possible in discrete-time.
As evidence of this claim, not only are complex dynamic phenomena prevalent in low dimensional discrete systems where it is impossible in continuous systems~(e.g.~chaos~\citep{Thip18}), but Figure~\ref{fig:shadowing} demonstrates the robustness of MWU by showing it can follow replicator dynamics on a matrix game derived by~\cite{andrade2021learning} in order to emulate the iconic Lorenz strange attractor.
In future work, instead of using continuous learning dynamics as a proxy, directly simulating Turing machines with discrete dynamics may provide powerful tools for learning in games.
Research on Turing machine simulations using physical systems has a rich history and encompasses far more than what is discussed in this paper.
Various techniques have been used to directly simulate Turing machines using discrete dynamics~\citep{moore1990unpredictability,siegelmann1992computational}, and insights from this prior work may hold potent insights for applications to learning in games.

\section*{Acknowledgments}
We thank Joshua Grochow for their insights, discussions, and references.
This research-project is supported in part by the National Research Foundation, Singapore under NRF 2018 Fellowship NRF-NRFF2018-07, AI Singapore Program (AISG Award No: AISG2-RP-2020-016), NRF2019-NRF-ANR095 ALIAS grant, AME Programmatic Fund (Grant No. A20H6b0151) from the Agency for Science, Technology and Research (A*STAR), grant PIE-SGP-AI-2018-01 and Provost's Chair Professorship grant RGEPPV2101.
This material is based upon work supported by the National Science Foundation under Grant No. IIS-2045347.

\bibliographystyle{plainnat}
\bibliography{refs}

\appendix

\section{Turing Complete Polynomial Flows on $\US^{17}$}\label{a:cardona}
We will briefly sketch the construction by~\cite{cardona2021turing} of the Turing complete polynomial vector field in Proposition~\ref{prop:cardona}; for a complete treatment we refer the reader to~\cite{cardona2021turing}.
To simplify notation, throughout this section we will represent a step in a Turing machine $T$'s evolution~(i.e.~an iteration of Steps 1--3 in~\S\ref{subsec:TM_prelim}) by the global transition function
\begin{equation*}
    G_T: Q \times A \to Q \times A~,
\end{equation*}
where we set $G_T(q_{\text{halt}}, w) \coloneqq (q_{\text{halt}}, w)$ for any tape $w$.

Let $T = \left(Q,\Sigma,\delta,q_0,q_{\text{halt}} \right)$ be a Turing machine.
We begin by encoding each configuration of $T$ as a constructible point in $\reals^3$. 
Let $r = |Q|$ be the cardinality of the set of states $Q$, then we will represent the elements of $Q$ by $[r] = \{1, \dots, r\} \in \N$.
Since we know tapes satisfy eq.~\eqref{eq:finite_tape}, we can encode any such $w = (w_i)_{i \in \Z}$ as the pair of points in $\N^2$ given by 
\begin{align*}
    y_1 &= w_0 + w_1 10 + \dots + w_{k_0} 10^{k_0}\\
    y_2 &= w_{-1} + w_{-2} 10 + \dots + w_{-k_0} 10^{k_0 - 1}~.
\end{align*}
Taken together, we have an encoding of every $(q,w) \in Q \times A$ as $(y_1,y_2,q) \in \N^3 \subset \reals^3$.
Define $\zeta: Q \times A \to \N^3$ as the map assigning each configuration in $Q \times A$ its associated point in $\N^3$ that we constructed.
The global transition function $G_T$ can now be reinterpreted as a map from $\zeta(Q \times A) \subset \N^3$ to itself.
By extending said map to be the identity map on points in $\N^3 \setminus \zeta(Q \times A)$, we arrive at a map on the whole of $\N^3$ to itself---for simplicity, we will denote this extended map by $G_{\zeta(T)}: \N^3 \to \N^3$.

Using this encoding, the next step in the construction is to simulate $T$ using a polynomial vector field $P$ on $\reals^{n+3}$.
To this end, a modification of a construction by~\cite{graca2008computability} is given.
Specifically,~\cite{graca2008computability} construct a non-autonomous polynomial vector field that simulates $T$, and this vector field is made autonomous via a standard trick of introducing a proxy variable in place of the explicit dependence on time.
Let $P$ on $\reals^{n+3}$ be the \emph{autonomous} polynomial vector field derived via this modification.
The construction by~\cite{graca2008computability} also shows how, given an input tape $s \in A$, a point $p_s = \left(\zeta(q_0,s), \tilde{y}_0\right) \in \reals^{n+3}$ is constructed so that the trajectory of $P$ starting from $p_s$ will simulate $G_{\zeta(T)}$.
The term $\zeta(q_0,s) \in \reals^3$ is defined above and the term $\tilde{y}_0 \in \reals^n$ is from a composition of polynomials depending only on $T$ and $s$---neither of these points are affected by the modification and can be taken as is.
The group property of flows ensures that any trajectory passing through $p_s$ is equivalent to a trajectory ending at and then ``restarting'' from $p_s$, so we can assume $p_s$ is an initial condition in Definition~\ref{def:turing_comp} without loss of generality.
Suppose we have a finite string $w^* = (w^*_{-k}, \dots, w^*_k)$ of symbols in $\Sigma$, we will now construct the set $U_{w^*}$ in Definition~\ref{def:turing_comp}.\footnote{For brevity we will brush over the construction of this set on the component corresponding to the proxy variable for time. Technically this component should be a union of small open intervals for each $i \in \N$, which intuitively associates a rough length of time in the dynamical system with a step in the Turing machine. However, formally introducing this portion of the construction is not particularly insightful since the relevance to the proof is rather tautological due to the proxy variable monotonically increasing at the same constant rate as time.}
Let $\omega = \{w \in \Sigma^{\Z}~|~w_i = w^*_i~\forall i \in [-k,k]\}$, $\epsilon > 0$ be a small positive constant, and $\restr{\reals^3}{\zeta(q_{\text{halt}}, \omega)}$ be the set of points in $\reals^3$ corresponding to configurations of $T$ of the form $(q_{\text{halt}},w \in \omega)$.
Defining $U_{w^*}^\epsilon \subset \reals^3$ as an $\epsilon$-neighborhood of $\restr{\reals^3}{\zeta(q_{\text{halt}}, \omega)}$ gives the open set
\begin{equation*}
    U_{w^*} \coloneqq U_{w^*}^\epsilon \times \reals^{n}~.
\end{equation*}
Showing $P$ satisfies Definition~\ref{def:turing_comp} with this choice of $p_s$ and $U_{w^*}$ follows from a relatively straightforward argument using properties inherited from the construction by~\cite{graca2008computability}.
Finally, the polynomial vector field $X$ in Proposition~\ref{prop:cardona} is constructed by using the pullback of inverse stereographic projection on a suitable reparametrization of $P$ and taking $T$ to be a universal Turing machine.
The pullback of inverse stereographic projection ensures that $X$ is a polynomial vector field tangent to the sphere and the reparametrization ensures the vector field is bounded.\footnote{Technically $X = \restr{\bar{X}}{\US^{n+4}}$, where $\bar{X}$ is a polynomial vector field on $\reals^{n+5}$ and tangent to $\US^{n+4}$. Similarly, as discussed in the proof of Theorem~$1.3$ by~\cite{cardona2021turing}, the reparametrization ensures the vector field is global because it is bounded.}
The fact that $X$ is well-defined on $\US^{17}$ and has degree $58$ follows from an analysis by~\cite{hainry2009decidability} of the construction by~\cite{graca2008computability}.

\section{Proof of Corollary~\ref{cor:bounding_degree}}\label{a:cor1}
\setcounter{theorem}{5}
\begin{corollary}
For some $m \leq {76 \choose 18}+1$, there exists a matrix game $A \in \reals^{m \times m}$ such that replicator dynamics on $A$ is Turing complete.
\end{corollary}
\begin{proof}
Let $X$, $\bar{X}$, and $Y$ be the vector fields defined in the proof of Theorem~\ref{thm:main}.
Similarly, let $A \in \reals^{m \times m}$ be the matrix game we arrived at by applying Proposition~\ref{prop:glv} to $Y$.
From Proposition~\ref{prop:glv} we know that $m-1$ is at least the number of unique monomials in the generalized polynomials in $Y$, so the proof follows by bounding the number of unique monomials from above.

From Proposition~\ref{prop:cardona} we know that $X$ is a polynomial vector field of degree $58$.
As mentioned in Appendix~\ref{a:cardona}, the specific degree of $58$ was derived from follow-up work by~\cite{hainry2009decidability} analyzing the construction by~\cite{graca2008computability}.
However, although the vector field is technically constructible, actually constructing $X$ to simulate a universal Turing machine is non-trivial in practice.
With this complication in mind, a crude upper bound on the number of unique monomials in $X$ is simply the number of unique monomials of degree $58$ in $18$ variables.
Therefore, a standard combinatorial argument tells us that the number of unique monomials in the polynomials of $X$ is at most ${58+18 \choose 18} = {76 \choose 18}$.
The construction of $\bar{X}$ cannot increase the number of monomials counted by this combinatorial argument since it can only introduce unique monomials via the term $1 - \|\x\|_2^2$, which is already counted in the bound ${76 \choose 18}$.
Similarly, we construct $Y$ by translating $\bar{X}$ by a constant and therefore can only introduce the constant monomial (i.e.~terms with all variables having zero exponents) which is already being counted.
Thus we have found that $m-1 \leq {76 \choose 18}$, which implies $m \leq {76 \choose 18} + 1$.
\end{proof}

\section{Deriving MWU as discrete-analogue of Replicator Dynamics}\label{a:MWU_to_rep}
Let $\delta: \reals^{n} \to \Delta^{n}$ be the \emph{logit map} defined as 
\begin{equation*}
    \delta_i(\y) = \frac{\ex{\y_{i}}}{\sum_{j \in [n]} \ex{\y_j}}~, \qquad \y \in \reals^{n}, i \in [n]~.
\end{equation*}
\cite{hofbauer2009time} showed that the flow generated by replicator dynamics can be written as 
\begin{equation}\label{eq:logit_rep}
    \x_i(t) = \delta_i(\y(t)) = \frac{\ex{\y_{i}(t)}}{\sum_{j \in [n]} \ex{\y_{j}(t)}}~, \qquad \y(0) \in \reals^{n}, i \in [n], t \in \reals~,
\end{equation}
where $\x$ and $\y$ are the mixed strategy and cumulative payoff vectors given in eq.~\eqref{eqn:FTRL}.
Rewriting eq.~\eqref{eq:logit_rep} in the form of eq.~\eqref{eqn:FTRL} gives an explicit representation of replicator dynamics' trajectories as a functions of cumulative payoffs,
\begin{equation}
\label{eq:rep_disc_prep}
\begin{aligned}
\y_i(t)&= \y_{i}(0)+\int_0^t \sum_{j \in [n]} A_{i,j}\delta_j(\y(s)) ds \\
\x_i(t)&= \delta_i(\y(t))~.
\end{aligned}
\end{equation}
By applying a standard Euler discretization with step size $\eta$ to the payoffs $\y$ in eq.~\eqref{eq:rep_disc_prep}, we find
\begin{equation*}
    \y_i(t + \eta) \approx \y_i(t) + \eta \ \dot{\y}_i(t) = \y_i(t) + \eta \sum_{j \in [n]} A_{i,j}\delta_j(\y(t))~.
\end{equation*}
Finally, iteratively applying this Euler discretization of the cumulative payoffs and using the logit map will give us the well-known MWU algorithm.
Formally, denoting the discretization's $t^{\text{th}}$ iterate by $\y^t$, we write MWU as
\begin{equation}
\label{eq:disc_mwu}
\begin{aligned}
\y_i^{t+1} &= \y_i^t + \eta \sum_{j \in [n]} A_{i,j}\delta_j(\y^{t}) = \y_i^0 + \eta \sum_{\tau = 1}^{t}A_{i,j}\delta_j(\y^{t})\\
\x_i^{t+1} &= \delta_i(\y^{t+1}) = \delta_i(\y^0 + \eta \sum_{\tau = 1}^{t}A_{i,j}\delta_j(\y^{t}))~.
\end{aligned}
\end{equation}

As the form of MWU in eq.~\eqref{eq:disc_mwu} was found via an Euler discretization, a standard result in numerical analysis tells us that the error accrued by a single iteration of MWU starting from the same initial conditions as replicator dynamics will satisfy
\begin{equation*}
    \|\y_i^{1} - \y_i(\eta)\| \leq \mathcal{O}(\eta^2)~.
\end{equation*}
However, since we are simulating Turing machines in the space of mixed strategies, we need error bounds on the probability simplex itself and not in the space of cumulative payoffs.

\section{Proof of Lemma~\ref{lem:LE}}\label{a:LE}
\setcounter{theorem}{9}
\begin{lemma}
Let $\Phi$ be the flow generated by replicator dynamics and $\x^{t}$ be the mixed strategy found on the $t^{\text{th}}$ iterate of MWU. 
The error accrued by a single iteration of MWU with step-size $\eta > 0$ is
\begin{equation*}
    \|\x^{t+1} - \Phi(\eta, \x^{t})\|_{\infty} \leq  1 - e^{-2\eta}~.
\end{equation*}
\end{lemma}
\begin{proof}
Suppose without loss of generality that for any action $i$ the expected payoff is bounded to $[-1,1]$, i.e.~$\sum_{j \in [n]} A_{i,j}\delta_j(\y) \in [-1,1]$.\footnote{The assumption that expected payoffs are bounded to $[-1,1]$ does not affect learning dynamics since we can always normalize the payoff matrix by its largest element.}
Let $W(t) = \sum_{j \in [n]} \ex{\y_{j}(t)}$ and $W_i(t) = \ex{\y_{i}(t)} = x_i(t) W(t)$.
Then continuous time RD becomes $$\x_i(t) = \frac{W_i(t)}{W(t)}~.$$
Similarly, define $\hat{W}^t = \sum_{j \in [n]} \ex{\y_{j}^{t-1}}$ and $\hat{W}_i^t = \ex{\y_{i}^{t-1}} = x_i^t \hat{W}^t$.
Then MWU becomes $$\x_i^t = \frac{\hat{W}_i^t}{\hat{W}^t}~.$$

We are interested in bounding the local error of MWU as a discretization of RD, i.e.~the error introduced by a single step of MWU relative to RD after a single step starting from the same point.
Thus without loss of generality we will focus on the first iterate of MWU and RD after $t = \eta$ amount of time.
Since expected payoffs are bounded to $[-1,1]$ we deduce from the analysis in Appendix~\ref{a:MWU_to_rep} that $$\hat{W}_i^1 \ex{-\eta} \leq W_i(\eta) \leq \hat{W}_i^1 \ex{\eta}~,$$
which implies $$\hat{W}^1 \ex{-\eta} \leq W(\eta) \leq \hat{W}^1 \ex{\eta}~.$$
Hence $$\x_i^1 \ex{-2\eta} \leq \x_i(\eta) \leq \x_i^1 \ex{2\eta}$$ whenever RD and MWU start from the same initial condition.

We have thus found that the local error introduced by a single time step is $$\|\x^{1} - \x(\eta)\| \leq  \|\x^{1} - \x^1 \ex{-2\eta}\| \leq  |1 - \ex{-2\eta}|~.$$
Observing that $\eta > 0$ gives the result.
\end{proof}

\section{Proof of Lemma~\ref{lem:GE}}\label{a:GE}
\setcounter{theorem}{10}
\begin{lemma}
Let $\Phi$ be the flow generated by replicator dynamics and $\x^{t}$ be the mixed strategy found on the $t^{\text{th}}$ iterate of MWU.
The error accrued after $t+1$ iterations of MWU with step-size $\eta > 0$ is
\begin{equation*}
    \|\x^{t+1} - \Phi(t\eta, \x^{0})\| \leq \mathcal{O}\left(e^{t}\right)~.
\end{equation*}
\end{lemma}
\begin{proof}
The flow $\Phi$ is $C^1$ and $\Delta^n$ is compact, so we know that $\Phi$ is Lipschitz continuous.
Let $L$ denote the Lipschitz constant for $\Phi$ with respect to $\|\cdot\|_{\infty}$.
It follows that for every initial condition $\x^0 \in \Delta^n$,
\begin{align*}
    E^{t+1} = \|\x^{t+1} - \Phi((t+1)\eta, \x^{0})\| &= \|\x^{t+1} - \Phi(\eta,\Phi(t\eta, \x^{0}))\|\\
    &= \|\x^{t+1} - \Phi(\eta,\x^{t}) + \Phi(\eta,\x^{t}) - \Phi(\eta,\Phi(t\eta, \x^{0}))\|\\
    &\leq \|\x^{t+1} - \Phi(\eta,\x^{t})\| + \|\Phi(\eta,\x^{t}) - \Phi(\eta,\Phi(t\eta, \x^{0}))\|\\
    &\leq |1 - \ex{-2\eta}| + e^{\eta L}\|\x^{t} - \Phi(t\eta, \x^{0})\|\\
    &= |1 - \ex{-2\eta}| + e^{\eta L} E^{t}
\end{align*}

To conclude our proof, we require a special case of the discrete Gronwall lemma.
This powerful tool for numerical error analysis tells us that if, for some constants $a$ and $b$ with $a>0$, a positive sequence $\{z^\tau\}_{\tau=0}^{t}$ satisfies 
\begin{equation*}
    z^{\tau+1} \leq b + a z^{\tau}~, \qquad \forall \tau \in [t-1]~,
\end{equation*}
then for $a \neq 1$
\begin{equation*}
    z^{\tau} \leq a^\tau z^0 + b \frac{a^\tau - 1}{a - 1}~, \qquad \forall \tau \in [t]~,
\end{equation*}
and for $a = 1$
\begin{equation*}
    z^{\tau} \leq z^0 + \tau b~, \qquad \forall \tau \in [t]~.
\end{equation*}

Recall that both $\eta > 0$ and $L > 0$ by definition. 
Let $z^{t} = E^{t}$, $a = e^{\eta L}$, and $b = |1 - \ex{-2\eta}|$.
Applying the discrete Gronwall lemma yields
\begin{equation*}
    E^{t+1} \leq e^{(t+1) \eta L} E^{0} + |1 - \ex{-2\eta}| \frac{e^{(t+1) \eta L} - 1}{e^{\eta L} - 1}~.
\end{equation*}
Clearly $E^0 = 0$ since MWU and replicator dynamics have the same initial conditions.
Thus we have shown
\begin{equation*}
    E^{t+1} \leq |1 - \ex{-2\eta}| \frac{e^{(t+1) \eta L} - 1}{e^{\eta L} - 1}~,
\end{equation*}
which concludes the proof.
\end{proof}

\end{document}